\documentclass[aps,12pt,tightenlines,nofootinbib]{revtex4}
\usepackage{amsmath,amssymb,amsthm}

\def\xnewpage{} 

\newtheorem{Theorem}{Theorem}

\newtheorem{Lemma}{Lemma}

\def\ra{\rightarrow}

\def\C{\mathcal{C}}

\def\R{\mathbb{R}}
\def\RTX{\R_+^{1+3}}
\def\RTR{\R_+^{1+1}}
\def\Linf{{L^\infty}}

\def\Linftx#1{{L^\infty_{#1}}} 
\def\Linftr#1{{\widehat{L}^\infty_{#1}}} 

\def\LinfTX{{L^\infty(\RTX)}}

\def\<{\langle}
\def\>{\rangle}
\def\norm#1{\<#1\>}

\def\d{\partial}

\def\eps{\varepsilon}
\def\la{\lambda}

\def\L{\Box^{-1}}
\def\O{O}
\newcommand{\OO}[2][]{O_{#1}\left(#2\right)}

\begin{document}

\title{Asymptotics from scaling for nonlinear wave equations}

\author{Nikodem Szpak}
\affiliation{Max-Planck-Institut f\"{u}r
Gravitationsphysik, Albert-Einstein-Institut, Golm, Germany}
\date{\today}

\begin{abstract}
  We present a scaling technique which transforms the evolution problem for a nonlinear wave equation with small initial data to a linear wave equation with a distributional source.
The exact solution of the latter uniformly approximates the late-time behavior of solutions of the nonlinear problem in timelike and null directions.
\end{abstract}


\maketitle

\section{Introduction}

For quite some time it is known that small global solutions of nonlinear wave equations
\begin{equation} \label{wave-eq}
   \Box u = F(u)
\end{equation}
develop late-time power-law tails. Several authors \cite{Asakura, John-blowup, Strauss-T, NS-Tails} have given pointwise decay estimates which imply bounds on the decay rates as $t\ra\infty$. However, they do not  show that these decay rates are optimal. It was only recently that we were able to calculate the exact decay rates and gave conditions on the initial data under which the decay estimates are optimal \cite{NS-Tails, NS-PB_Tails}. Our proof was, however, restricted to the asymptotics in the timelike direction, i.e. $t\ra\infty$ with $x$ being fixed. Here, we give a uniform asymptotic formula which holds in the timelike and null infinity, i.e. for $t\ra\infty$ and $0<|x|/t<1$.

We consider the initial value problem for nonlinear wave equations of the form \eqref{wave-eq} with small initial data
\begin{equation} \label{init-data}
  (u, \d_t u)|_{t=0} = (\eps f,\eps g)
\end{equation}
in three spatial dimensions, 
i.e. $u:(t,x)\in \R_+\times\R^3\ra\R$,
where $\eps$ is a small number. The nonlinearity behaves like
\begin{equation} \label{nonlinearity}
 F(u) = u^p + O(u^{p+k})\qquad \text{for }u\ra 0
\end{equation}
with $p>1+\sqrt{2}$ and $k>0$ (if $p$ is non-integer the expression $u^p$ is to be understood as $|u|^{p-1}u$ everywhere where needed).
The main idea developed in this article is based on the observation that in the limit $\eps\ra 0$ the solutions $u_\eps$ tend under suitable scaling to some nontrivial $u_*$ which satisfies a linear wave equation with a distributional source. This equation can be solved exactly.
We show that for small but finite values of $\eps$ the solutions $u_\eps$ are near to $u_*$, in a suitable sense with a uniform error bound, such that $u_*$ determines the late-time asymptotics of $u_\eps$ and plays a role of a universal attractor in the evolution.
Our main result is the asymptotic formula
\begin{equation}
  u(t,r) = \eps u_0(t,r) +
  \eps^p \frac{A_p}{r} \left[ \frac{1}{(t-r)^{p-2}} - \frac{1}{(t+r)^{p-2}} \right]
  + \OO[p-1]{\eps^{p+\la}}
\end{equation}
where $u_0$ is a solution of the linear problem with data \eqref{init-data} and the second term represents the leading nonlinear asymptotics with $A_p$ being
the only trace of the initial data.
The error term $\O_{p-1}$ (to be made precise below) is uniformly small relative to the leading asymptotics as $\eps\ra 0$ and $\la$ is some positive number depending on $p$ and $k$.

This asymptotic formula recovers the known timelike asymptotics $u\sim \eps^p/t^{p-1}$ \cite{NS-PB_Tails} as well as null asymptotics $u\sim 1/t$ and is regular at $r=0$.

The idea of scaling has been inherited from the Doctoral Thesis of Hans Lindblad \cite{Lindblad-PhD_CPDE} where he, in contrast, used similar scaling technique to treat blowup in finite time of solutions to the wave equations $\Box u = u^p$ with $p<1+\sqrt{2}$.

The scaling method for small global solutions introduced here can also be straightforwardly extended to more complicated systems of nonlinear wave equations as long as appropriate a priori decay estimates are known. However, even if no such estimates are available this technique can act as a (nonrigorous) guide to the late-time asymptotics of the studied system.
This method can be also used to study the sub-leading asymptotics, order by order in an iterative way, which we plan to describe in a subsequent publication.

For more transparency in this introductory paper we restrict ourselves to spherical symmetry as being known to dominate the late time asymptotics of the studied equation anyway. For generalization to full 3-dimensions, which here is a purely technical element, one can repeat the corresponding steps from Lindblad \cite{Lindblad-PhD_CPDE}.

This paper is organized as follows. First we introduce the idea of scaling and discuss the limit $\eps\ra 0$ in a non-rigorous way. Then we prove a theorem on late-time asymptotics of \eqref{wave-eq} by comparing the original system with the limiting rescaled one.
Finally, we discuss some applications as well as the extensions of the technique.

\xnewpage
\subsection*{Notation}


With the symbol $\<x\>:=1+|x|$ we define a space-time weighted-$L^\infty$ norm
\begin{equation}
  \| u \|_\Linftx{p} := \| \norm{t+|x|} \norm{t-|x|}^{p-1}  u(t,x)\|_\LinfTX.
\end{equation}
Its finiteness guarantees the decay of $u$ like $1/t$ on the lightcone $t\sim |x|$ and like $1/t^p$ for fixed $x$ as well as $1/|x|^p$ for fixed $t$.
We also define a smaller norm restricted to the region $\Omega:=\{(t,x)\in \RTX: t-|x|>1\}$
\begin{equation}
  \| u \|_\Linftr{p} := \| \norm{t+|x|} \norm{t-|x|}^{p-1}  u(t,x)\|_{\Linf(\Omega)}
  \leq \| u \|_\Linftx{p}.
\end{equation}

In 3 spatial dimensions by $\RTX$ we will understand $\R_+\times \R^3$ and in spherical symmetry $\RTR$ will mean $\R_+\times \R_+$ with $\{0\}\in\R_+$.

By $\L$ we will denote the inverse of the wave operator $\Box$ having the properties
\begin{align}
  u &= \L F& &:\Leftrightarrow & \Box u &= F& &\text{and}& (u,\d_t u)|_{t=0}&=(0,0)
\end{align}
In 3 spatial dimensions $\L$ is a positive measure what gives a useful comparison property
\begin{equation}\label{Box-positive}
  F \geq G \qquad \Rightarrow \qquad \L F \geq \L G.
\end{equation}

We define the symbol $\O_q$ as an extension of the \textit{big-O} symbol for $\eps\ra 0$ to the space $\Linftx{q}$, i.e. for a family $F_\eps\in\Linftx{q}$ and $G\in\C(\R)$, $g(0)=0$ we say
\begin{equation} \label{def:O}
  F_\eps=\O_q(G(\eps))\qquad :\Leftrightarrow \qquad \|F_\eps\|_\Linftx{q} = O(G(\eps)).
\end{equation}
Whenever we write $\L \O(F_\eps)$ or $\L \O_{p}(\eps^q)$ we mean $\L$ is acting on a function of prescribed decay which is at least in $\C^2(\RTX)$.

\xnewpage
\section{The idea of scaling}

In this section we want to sketch our strategy of using scaling to obtain exact late time asymptotics for small data. To this end we consider a simplified version of \eqref{wave-eq}
\begin{equation}
  \Box u = u^p
\end{equation}
with integer $p\geq 3$ and initial data
\begin{equation}
  u(0,r)=\eps f(r),\qquad \d_t u(0,r)=\eps g(r)
\end{equation}
being smooth functions of compact support in $[0,R]$.
First, we solve the linear equation with removed scale factor $\eps$
\begin{equation} \label{free-wave}
  \Box u_0 = 0,
\end{equation}
\begin{equation} \label{free-wave-init}
  u_0(0,r)=f(r),\qquad \d_t u_0(0,r)=g(r).
\end{equation}
Its solution can be written in the form
\begin{equation} \label{free-u0}
  u_0(t,r)=\frac{h(t-r)-h(t+r)}{r}
\end{equation}
where
\begin{equation} \label{def:h}
  h(x):=-\frac{x}{2} f(x) - \frac{1}{2} \int_x^\infty y g(y) dy
\end{equation}
has support in $[-R,R]$ (the functions $f(r), g(r)$ have been symmetrically continued to negative $r$).
Next, we subtract the linear solution from the nonlinear by introducing
\begin{equation}
  w(t,r):=u(t,r) - \eps u_0(t,r)
\end{equation}
which satisfies
\begin{equation}
  \Box w = (w + \eps u_0)^p.
\end{equation}
Now, we scale this function to
\begin{equation} \label{def-Weps}
  W_\eps(t,r):=\eps^{-b} w(\eps^{-a} t, \eps^{-a} r)
\end{equation}
with $b=p+a(p-1)$ and some $a>0$ to be chosen later. It satisfies
\begin{equation} \label{free-W_eps}
\begin{split}
  \Box W_\eps(t,r) &= \eps^{-a} \left[\eps^{(p-1)+a(p-2)} W_\eps(t,r) + \eps^{-a} u_0(\eps^{-a} t, \eps^{-a} r) \right]^p \\
  &= \eps^{-a} \left[ \eps^{-a} u_0(\eps^{-a} t, \eps^{-a} r) \right]^{p} \\
  &+ \sum_{m=1}^p \begin{pmatrix} p\\ m\end{pmatrix} \eps^{-a} \left[ \eps^{-a} u_0(\eps^{-a} t, \eps^{-a} r) \right]^{p-m} \eps^{[(p-1)+a(p-2)]m} W_\eps^m(t,r).
\end{split}
\end{equation}
For this equation we want to consider the limit $\eps\ra 0$. From our previous work \cite{NS-WaveDecay} we have global weighted bounds on $u$ and $w$ when $\eps$ is sufficiently small
\begin{align}
  |u(t,r)| &\leq \frac{C \eps}{\norm{t+r}\norm{t-r}^{p-2}}, &
  |w(t,r)| &\leq \frac{C \eps^p}{\norm{t+r}\norm{t-r}^{p-2}}.
\end{align}
It implies that the functions $u, w$ and by \eqref{def-Weps} also $W_\eps(t,r)$ are bounded in the weighted-$L^\infty$ norm
\begin{align}
  \|u\|_\Linftx{p-1} &\leq C \eps, & \|w\|_\Linftx{p-1} &\leq C \eps^p,
  & \|W_\eps\|_\Linftx{p-1} &\leq C\eps^{-a(p-1)}, & \|W_\eps\|_\Linftr{p-1} &\leq C.
\end{align}

Let us recall the following fact of the distributional calculus: any smooth function $H(x)$ of compact support can be squeezed to the delta distribution under appropriate scaling as $\eps\ra 0$,
\begin{align}
  \eps^{-1} H(\eps^{-1} x) &\ra C_H\ \delta(x), & C_H &:= \int H(x)\ dx.
\end{align}
Having this in mind we observe that all terms in \eqref{free-W_eps} containing powers of $u_0$, by use of the representation \eqref{free-u0}, will have a distributional limit 
\begin{equation} \label{u0-delta}
  \eps^{-a} \left[ \eps^{-a} u_0(\eps^{-a} t, \eps^{-a} r) \right]^{p-m} \ra
  C_{p-m} \frac{\delta(t-r)-\delta(t+r)}{r^{p-m}}
\end{equation}
with
\begin{equation} \label{def:Ck}
  C_m := \int h^m(x) dx.
\end{equation}
The term $\delta(t+r)$ will further play no role since its support is outside of the region of our interest $t+r>0$.
Applying these limits to equation \eqref{free-W_eps} we observe that the first term tends to
\begin{equation}
  C_p \frac{\delta(t-r)}{r^p}
\end{equation}
while the sum over $m$ vanishes as $\eps\ra 0$ because it converges to the delta distributions times functions $W^m_\eps(t,r)$ which are uniformly bounded by $C \eps^{-a(p-1)m}$ and factors $\eps^{[(p-1)+a(p-2)]m}$, combined together giving
\begin{equation}
  \eps^{[(p-1)(1-a)+a(p-2)]m} \ra 0
\end{equation}
for any $p>2$ and $0<a<p-1$. This leads us to a limiting equation
\begin{equation}
  \Box W(t,r) = C_p \frac{\delta(t-r)}{r^p}
\end{equation}
which can be solved exactly
\begin{equation}
  W(t,r) = A_p \frac{\Theta(t-r)}{r} \left[ \frac{1}{(t-r)^{p-2}} - \frac{1}{(t+r)^{p-2}} \right]
\end{equation}
with $A_p:= 2^{p-3} C_p / (p-2)$.
After $W$ is traced back to the original solution $u$ it will provide the leading asymptotics for small $\eps$. In the rigorous treatment all error terms which we ignored here must be estimated as uniformly small is some region of spacetime. The crucial point will be to show that $W_\eps$ tends to $W$ in the norm $\Linftr{p-1}$, i.e.
\begin{equation}
  \lim_{\eps\ra 0} \|W_\eps - W\|_\Linftr{p-1} = 0,
\end{equation}
hence
\begin{equation}
  W_\eps(t,r) = W(t,r) + \O_{p-1}(1)
\end{equation}
uniformly in $t-r>1$.
Then we can reconstruct the asymptotics of $u$
\begin{equation}
\begin{split}
  u(t,r) &= \eps u_0(t,r) + \eps^{p+a(p-1)} W_\eps( \eps^a t, \eps^a r) \\
  &= \eps u_0(t,r) +
  \eps^p A_p \frac{\Theta(t-r)}{r} \left[ \frac{1}{(t-r)^{p-2}} - \frac{1}{(t+r)^{p-2}} \right]
  + \O_{p-1}(\eps^{p})
\end{split}
\end{equation}
which holds uniformly (in the sense of the $\Linftx{p-1}$-norm) in $t-r>\eps^{-a}$.

This asymptotic formula is richer than the one we obtained in \cite{NS-PB_Tails} for it additionally gives an analytic dependence on $r$ in the whole region $t-r>\eps^{-a}$. Hence, it can be thought of as asymptotics for $t\ra\infty$ including dependence on the parameter $r/t$ in the range $[0,1)$. The asymptotic formula of \cite{NS-PB_Tails} is reproduced by fixing $r$ and letting $t\ra\infty$ what gives
\begin{equation}
  u(t,r) \sim \eps^p/t^{p-1}
\end{equation}
because $u_0$ is localized in the vicinity of the outgoing lightcone $t=r$ and vanishes for $t\gg r$ leaving only the contribution of $W$.

Here, essential is that the error term represented by $\O_{p-1}(\eps^{p})$ is controlled in the $\Linftx{p-1}$ norm, because only this guarantees that the error has at least the same falloff as $W$ while having a smaller amplitude for $\eps\ra 0$.

\xnewpage
\section{Main theorem}
In this section we proceed with allowing general nonlinearities of the form \eqref{nonlinearity} with any real $p>1+\sqrt{2}$ (not necessarily integer). The initial data will still be of compact support $[0,R]$, although data of sufficiently fast fall-off at infinity could be also considered within this scheme. This would, however, lead to somewhat more complicated asymptotics as an incoming radiation may alter the solution at late times (for more details on competition of the asymptotic effects, see \cite{NS-Tails}).
We choose $(f,g)\in\C^3\times\C^2$ such that the following norms are finite
\begin{align} \label{init-data-bounds}
  f_0 &:= \|\norm{r}^{p-1} f(r)\|_\Linf,&
  f_1 &:= \|\norm{r}^{p} \nabla f(r)\|_\Linf, &
  g_0 &:= \|\norm{r}^{p} g(r)|_\Linf.
\end{align}
By the results of \cite{NS-Tails}, it guarantees existence of a classical solution $u\in\C^2(\RTR)\cap\Linftx{p-1}$ for small values of $\eps$ which multiplies the initial data in \eqref{init-data}. The regularity can be relaxed to weak solutions in the energy space $u\in\C^0(\R,H^1)\cap\C^1(\R,L^2)$, but the estimates will then become much more delicate and the pointwise asymptotics only true \textit{almost everywhere} (cf. \cite{NS-WaveDecay} for such analysis).

Since we are interested in the late-time behavior of the solutions we must prescribe the asymptotic region in which our approximation should hold as $\eps\ra 0$. For there is some freedom in doing this a free \textit{scaling} parameter $a$ will be introduced. The quality of the error bound will then depend on its value, i.e. on the rate the asymptotic region shrinks to the vicinity of infinity when $\eps\ra 0$.

\begin{Theorem} \label{Th:main}
The solution $u$ of the initial problem \eqref{wave-eq}-\eqref{init-data} has the following asymptotic behavior for $\eps\ra 0$ and a given scaling parameter $0<a<p(p-1)/(p+1)$
\begin{equation}
  u(t,r) = \eps u_0(t,r) +
  \eps^p \frac{A_p}{r} \left[ \frac{1}{(t-r)^{p-2}} - \frac{1}{(t+r)^{p-2}} \right]
  + \OO[p-1]{\eps^{p+\la}}
\end{equation}
where $u_0$ solves the linear problem \eqref{free-wave}-\eqref{free-wave-init}, $A_p:= 2^{p-3} C_p / (p-2)$ with $C_p$ defined via \eqref{def:Ck} and \eqref{def:h} and $\la:=\min\{[p(p-1)(1-a)+a((p-1)^2-2)]/p, k(1+a), a\}$.
The error term is defined in \eqref{def:O} and  means that the asymptotics holds w.r.t. to the weighted-$\Linftx{p-1}$ norm. Here, it is restricted to the region $t-r>1/\eps^a$ what implies a uniform convergence there.
\end{Theorem}

Before we prove the Theorem, we first cite a useful lemma
\begin{Lemma} \label{Lem:a-b}
  If $p>1$ then for any real numbers $a,b$
  \begin{equation}
    |a+b|^p\leq 2^{p-1}(|a|^p+|b|^p)
  \end{equation}
  and
  \begin{equation}
    \Big| (a+b)^p - a^p - b^p \Big| \leq 2^{p-1} p \Big( |a||b|^{p-1} + |b||a|^{p-1} \Big).
  \end{equation}
The inequality remains true when all expressions of the type $x^p$ on the left-hand side are replaced by $|x|^p$ or $|x|^{p-1}x$.
\end{Lemma}
\begin{proof}
  See e.g. Lindblad \cite[Lemma 8.3]{Lindblad-PhD_CPDE} with minor modifications not affecting the result.
\end{proof}

\begin{proof}[Proof of Theorem \ref{Th:main}]
From \cite[Th.2]{NS-Tails} we know that there exists a classical solution $u\in\C^2(\RTR)$
which belongs to the $\Linftx{p-1}$ space.

Let $u_0$ solve the linear equation \eqref{free-wave} with the initial data \eqref{free-wave-init} satisfying \eqref{init-data-bounds}.
Then, by \cite[Lem. 4a]{NS-WaveDecay}, there exists a classical solution $u_0\in\C^2(\RTR)$
which satisfies
\begin{equation} \label{u0-norm}
  \|u_0\|_\Linftx{p-1} \leq C\cdot (g_0+f_1+f_0).
\end{equation}
(This can be also directly verified by using the explicit form \eqref{free-u0} of $u_0$.)
Define
\begin{equation}
  w(t,r):=u(t,r) - \eps u_0(t,r)
\end{equation}
and
\begin{equation}
  W_\eps(t,r):=\eps^{-b} w(\eps^{-a} t, \eps^{-a} r)
\end{equation}
with $b=p+a(p-1)$.
From Theorem 3.1 and inequality (3.7) of \cite{NS-WaveDecay} it follows\footnote{Notice a slightly different notation: $u_0$ here corresponds to $u_1$ in \cite{NS-WaveDecay}.} that
\begin{align}
  \|u\|_\Linftx{p-1} &\leq C \eps,& \|w\|_\Linftx{p-1} &\leq C \eps^p.
\end{align}
It implies
\begin{equation} \label{norm-Weps}
  \|W_\eps\|_\Linftx{p-1} \leq C\eps^{-a(p-1)}
\end{equation}
and
\begin{equation}
\begin{split}
  \|W_\eps\|_\Linftr{p-1} &= \eps^{-[p+a(p-1)]} \|\<t+r\>\<t-r\>^{p-2} w(t/\eps^a,r/\eps^a)\|_{\Linf(t-r>1)} \\
  &= \eps^{-[p+a(p-1)]} \left\|
     \frac{\<t+r\>\<t-r\>^{p-2}\cdot
     \<t/\eps^a+r/\eps^a\>\<t/\eps^a-r/\eps^a\>^{p-2}|w(t/\eps^a,r/\eps^a)|}
     {\<t/\eps^a+r/\eps^a\>\<t/\eps^a-r/\eps^a\>^{p-2}} \right\|_{\Linf(t-r>1)} \\
  &\leq \eps^{-[p+a(p-1)]} \|w\|_\Linftx{p-1} \left\|
     \frac{\eps^{a(p-1)}(1+t+r)(1+t-r)^{p-2}}
     {(\eps^a+t+r)(\eps^a+t-r)^{p-2}} \right\|_{\Linf(t-r>1)} \\
  &\leq \eps^{-[p+a(p-1)]} \|w\|_\Linftx{p-1} 2^{p-1} \eps^{a(p-1)} \\
  &\leq 2^{p-1} C,
\end{split}
\end{equation}
where we have used the inequality $1+t\pm r< 2(\eps^a+t\pm r)$ which holds for $t- r>1$.

$W_\eps$ satisfies the following nonlinear wave equation
\begin{equation} \label{BoxWeps}
  \Box W_\eps(t,r)
  = \eps^{-p-a(p+1)} F\left(\eps^{p+a(p-1)} W_\eps(t,r) + \eps u_0(\eps^{-a} t, \eps^{-a} r) \right)
\end{equation}
Now, lead by the observations made in the previous section, we introduce the \textit{limiting} function $W(t,r)$ defined by the \textit{limiting} equation
\begin{equation} \label{BoxW}
  \Box W(t,r) = C_p \frac{\delta(t-r)}{r^p} =: \mu,
\end{equation}
where
\begin{equation}
  C_k := \int h^k(x) dx
\end{equation}
and show that indeed $W_\eps$ converges to $W$ in norm $\Linftr{p-1}$, i.e.
\begin{equation}
  \lim_{\eps\ra 0} \|W_\eps - W\|_\Linftr{p-1} = 0.
\end{equation}
Note, that it cannot be shown that the right-hand side of \eqref{BoxWeps} converges in $\Linftx{p-1}$-norm to the right-hand side of \eqref{BoxW} as the norm of the latter is infinite. Convergence of $W_\eps$ to $W$, after inverting $\Box$, does not work either, since the $\Linftx{p-1}$-norm of $W$ is infinite. However, both $W_\eps$ and $W$ have finite $\Linftr{p-1}$-norms. This reflects the fact that pointwise convergence happens in the region $t-r>1$ but not in the vicinity of the light-cone $t\approx r$.

Let us consider their difference and split it into four parts for further analysis
\begin{equation}
\begin{split}
  |W_\eps - W|(t,r) &=
  \left|\L \eps^{-p-a(p+1)} F\left(\eps^{p+a(p-1)} W_\eps(t,r) + \eps u_0(\eps^{-a} t, \eps^{-a} r) \right)  - \L \mu \right| \\
&\leq \eps^{-a} \Big| \L  \Big(\eps^{(p-1)+a(p-2)} W_\eps(t,r) + \eps^{-a} u_0(\eps^{-a} t, \eps^{-a} r) \Big)^p  \\
  & \hspace{0.15\linewidth} - \Big( \eps^{[(p-1)+a(p-2)]} W_\eps(t,r)\Big)^p
  - \Big(\eps^{-a} u_0(\eps^{-a} t, \eps^{-a} r) \Big)^p \Big| \\
  &+ \eps^{-p-a(p+1)} \Big|\L  O\Big(\left[\eps^{p+a(p-1)} W_\eps(t,r) + \eps u_0(\eps^{-a} t, \eps^{-a} r)\right]^{p+k} \Big) \Big| \\
  &+ \eps^{p(p-1)+a[(p-1)^2-2]} \left| \L  W_\eps^p(t,r) \right|\\
  &+ \left| \L \eps^{-a (p+1)} u_0^p(\eps^{-a} t, \eps^{-a} r) - \L \mu \right| \\
  &=: \Delta_p + \Delta_O + \Delta_W + \Delta_\mu.
\end{split}
\end{equation}
The first term, $\Delta_p$, will be first estimated algebraically with Lemma \ref{Lem:a-b} and then in another Lemma, by H{\"o}lder inequality, reduced to a product of two separate expressions containing $\L W_\eps^p$ and $\L u_0^p$. Then two next Lemmas will transform these terms to norms of $W_\eps$ and $u_0$. Finally, use will be made of the norm bounds \eqref{norm-Weps} and \eqref{u0-norm}.
The second term, $\Delta_O$ will be estimated similarly.
The third term, $\Delta_W$, will be estimated by the same means, using norm bound \eqref{norm-Weps} only.
Finally, the estimate of the last term, $\Delta_\mu$, which is the core of the technique, will be based on the fact that $u_0$ scales to a distribution, as observed in the previous section (cf. \eqref{u0-delta}).

Lemma \ref{Lem:a-b} together with the positivity of $\L$ (cf. \eqref{Box-positive}) give
\begin{equation}
\begin{split}
  \Delta_p &\leq \eps^{-a} \L \Big| \Big(\eps^{(p-1)+a(p-2)} W_\eps(t,r) + \eps^{-a} u_0(\eps^{-a} t, \eps^{-a} r) \Big)^p  \\
  & \hspace{0.15\linewidth} - \Big( \eps^{[(p-1)+a(p-2)]} W_\eps(t,r)\Big)^p
  - \Big(\eps^{-a} u_0(\eps^{-a} t, \eps^{-a} r) \Big)^p \Big| \\
  &\leq p 2^{p-1} \eps^{-a} \L \Big(
  \Big|\eps^{[(p-1)+a(p-2)]} W_\eps(t,r)\Big| \Big|\eps^{-a} u_0(\eps^{-a} t, \eps^{-a} r) \Big|^{p-1} \\
  & \hspace{0.2\linewidth}+
  \Big|\eps^{[(p-1)+a(p-2)]} W_\eps(t,r)\Big|^{p-1} \Big|\eps^{-a} u_0(\eps^{-a} t, \eps^{-a} r) \Big|
  \Big)\\
  &= p 2^{p-1} \Big[ \eps^{p-1-2a} \L \Big(
  \Big| W_\eps(t,r)\Big| \Big| u_0(\eps^{-a} t, \eps^{-a} r) \Big|^{p-1} \Big) \\
  & \hspace{0.1\linewidth}+ \eps^{(p-1)^2+a p(p-3)} \L \Big(
  \Big| W_\eps(t,r)\Big|^{p-1} \Big|\eps^{-a} u_0( t, \eps^{-a} r) \Big|
  \Big) \Big]\\
\end{split}
\end{equation}
Now, we are going to separate the products under $\L$.
\begin{Lemma}
  For positive $\alpha, \beta, p>0$ such that $\alpha+\beta=1$ we have
  \begin{equation}
    |\L (f^{\alpha p} g^{\beta p})| \leq |\L f^p|^\alpha \cdot |\L g^p|^\beta
  \end{equation}
\end{Lemma}
\begin{proof}
  The proof is a straightforward calculation based on the fact that $\L f$ can be written as a convolution $E\ast f$, where $E$ is the fundamental solution of $\Box$, and use of the H{\"o}lder inequality with powers $1/\alpha, 1/\beta$ (cf. again Lindblad \cite[Lemma 8.11]{Lindblad-PhD_CPDE}).
\end{proof}
We apply this Lemma to our estimate on $\Delta_p$ above with $f=|W_\eps|, g=|u_0|$ and $\alpha=1/p, \beta=(p-1)/p$ or \textit{vice versa}. Then we get
\begin{equation}
\begin{split}
  \Delta_p
  &\leq p 2^{p-1} \Big[ \eps^{p-1-2a} \Big(\L\big| W_\eps^p(t,r)\big|\Big)^{1/p}
  \Big(\L\big| u_0^p(\eps^{-a} t, \eps^{-a} r) \big|\Big)^{(p-1)/p} \\
  & \hspace{0.1\linewidth}+ \eps^{(p-1)^2+a p(p-3)} \Big(\L\big| W_\eps^p(t,r)\big|\Big)^{(p-1)/p}
  \Big(\L\big| u_0^p(\eps^{-a} t, \eps^{-a} r) \big|\Big)^{1/p}
  \Big].
\end{split}
\end{equation}
It is convenient now to switch from pointwise to norm estimates. We get
\begin{equation} \label{W-Weps-norm}
  \|W_\eps - W\|_\Linftr{p-1} \leq
  \|\Delta_p\|_\Linftr{p-1} + \|\Delta_O\|_\Linftr{p-1}
  + \|\Delta_W\|_\Linftr{p-1} + \|\Delta_\mu\|_\Linftr{p-1}.
\end{equation}
These weighted-$\Linftx{q}$ norms have the following property (cf. Asakura \cite[Eq. 2.32]{Asakura}
\begin{Lemma} \label{Lem:product}
Let $v,w\in\Linftx{q}$ for some $q>1$. Then for any positive $\alpha, \beta>0$ such that $\alpha+\beta=1$ we have
\begin{equation*}
  \|v^{\alpha} w^{\beta}\|_\Linftx{q} \leq \|v\|_\Linftx{q}^\alpha \|w\|_\Linftx{q}^\beta.
\end{equation*}
The same holds for the $\Linftr{q}$ norms.
\end{Lemma}
By this Lemma with $q=p-1$ we get
\begin{equation}
\begin{split}
  \|\Delta_p\|_\Linftr{p-1}
  &\leq p 2^{p-1} \Big[ \eps^{p-1-2a} \Big\|\L\big| W_\eps^p\big|\Big\|_\Linftr{p-1}^{1/p}
  \Big\|\L\big| u_0^p(\eps^{-a} t, \eps^{-a} r) \big|\Big\|_\Linftr{p-1}^{(p-1)/p} \\
  & \hspace{0.05\linewidth}
  + \eps^{(p-1)^2+a p(p-3)} \Big\|\L\big| W_\eps^p\big|\Big\|_\Linftr{p-1}^{(p-1)/p}
  \Big\|\L\big| u_0^p(\eps^{-a} t, \eps^{-a} r) \big|\Big\|_\Linftr{p-1}^{1/p}
  \Big].
\end{split}
\end{equation}
At this stage we need two estimates for the two different terms containing $W_\eps$ and $u_0$, respectively. The first
we cite from Asakura \cite[Cor. 2.4 and Eq. 2.33]{Asakura}
\begin{Lemma} \label{Lem:power-p}
Let $v\in\C^2(\RTR)\cap\Linftx{q}$ for some $q>1$. Then for any $p>1+\sqrt{2}$
\begin{equation*}
  \|\L |v|^p\|_\Linftx{q} \leq C \|v\|_\Linftx{q}^p
\end{equation*}
with some $C>0$ provided $q\leq p-1$.
\end{Lemma}
It holds also when $\Linftx{q}$ is replaced by $\Linftr{q}$ on the left-hand side (but not on the right-hand side!). The proof is almost identical.

The second Lemma, using explicitly the scaling properties, we prove below
\begin{Lemma} \label{Lem:scaling-Lu_0}
For $u_0$ given by \eqref{free-u0} and for any $q>1$
\begin{equation}
  \Big\|\L\big| u_0^q(\eps^{-a} t, \eps^{-a} r) \big|\Big\|_\Linftr{q-1} \leq C \eps^{a(q+1)}
\end{equation}
with some $C>0$.
\end{Lemma}
\begin{proof}
Observe that due to scaling properties of the operator $\L$ we have
\begin{equation}\label{L-scaling}
  \big[\L u_0^q(\eps^{-a} \cdot, \eps^{-a} \cdot)\big](t,r) =
  \eps^{2a}\big[\L u_0^q\big](\eps^{-a} t, \eps^{-a} r).
\end{equation}
From Lemma \ref{Lem:power-p} and the bound \eqref{u0-norm} we get $\|\L |u_0|^q\|_\Linftr{q-1} \leq C \|u_0\|_\Linftx{q-1}^q \leq C'$. It gives the pointwise estimate
\begin{equation}
  \Big|\L |u_0^q|\Big|(t,r) \leq \frac{C}{\norm{t+r}\norm{t-r}^{q-2}}
\end{equation}
for $t-r>1$ which scales to
\begin{equation}\label{estimate-scaling}
  \Big|\L |u_0^q|\Big|(\eps^{-a} t, \eps^{-a} r)
  \leq \frac{C \eps^{a(q-1)}}{(\eps^a+t+r)(\eps^a+|t-r|)^{q-2}}
  \leq \frac{C' \eps^{a(q-1)}}{\norm{t+r}\norm{t-r}^{q-2}}
\end{equation}
for $t\pm r>1$.
Combining \eqref{L-scaling} with \eqref{estimate-scaling} we get the thesis.
\end{proof}

Now, Lemma \ref{Lem:power-p} and the bound \eqref{norm-Weps} give
\begin{equation}
  \Big\|\L\big| W_\eps^p\big|\Big\|_\Linftr{p-1} \leq
  \| W_\eps\|_\Linftx{p-1}^p \leq C \eps^{-a(p-1)p}
\end{equation}
while Lemma \ref{Lem:scaling-Lu_0} gives
\begin{equation}
  \Big\|\L\big| u_0^p(\eps^{-a} t, \eps^{-a} r) \big|\Big\|_\Linftr{p-1} \leq C \eps^{a(p+1)}.
\end{equation}
Finally, we arrive at
\begin{equation}
  \|\Delta_p\|_\Linftr{p-1}
  \leq C \Big[ \eps^{(p-1)(1-a)+a[(p-1)^2-2]/p} + \eps^{(p-1)\{(p-1)(1-a)+a[(p-1)^2-2]/p\}} \Big]
\end{equation}
with some constant $C>0$. Since we assumed $p>1+\sqrt{2}$ it holds $(p-1)^2-2>0$ and both $\eps$'s are raised to a positive power ($>\sqrt{2}$) for any $0<a<p(p-1)/(p+1)$.
The power of the second term is always bigger than that of the first term for their ratio is $p-1>1$.

For the $\Delta_O$ term we again use positivity of $\L$ and Lemma \ref{Lem:a-b} to get
\begin{equation}
\begin{split}
  |\Delta_O|
  &\leq 2^{p+k-1}  \L \O_{p-1}\Big(\eps^{p(p-1)+a[(p-1)^2-2]+[p+a(p-1)]k}|W_\eps(t,r)|^{p+k}\\
  &\hspace{0.2\linewidth}+ \eps^{-a(p+1)+k} | u_0(\eps^{-a} t, \eps^{-a} r)|^{p+k} \Big)
\end{split}
\end{equation}
and after taking the norm
\begin{equation}
\begin{split}
  \|\Delta_O\|_\Linftr{p-1}
  &= \O\left(\eps^{p(p-1)+a[(p-1)^2-2]+[p+a(p-1)]k} \|\L |W_\eps(t,r)|^{p+k}\|_\Linftr{p-1} \right)\\
  &+ \O\left(\eps^{-a(p+1)+k} \|\L | u_0(\eps^{-a} t, \eps^{-a} r)|^{p+k}\|_\Linftr{p-1} \right).
\end{split}
\end{equation}
Using the results obtained above we get
\begin{equation}
\begin{split}
  \|\Delta_O\|_\Linftr{p-1}
  &= O\left( \eps^{p(p-1)(1-a)+a[(p-1)^2-2]+kp} \right) + O\left(\eps^{k(1+a)} \right).
\end{split}
\end{equation}

It is now straightforward to estimate the $\Delta_W$ term. By Lemma \ref{Lem:power-p} we have
\begin{equation}
\begin{split}
  \|\Delta_W\|_\Linftr{p-1}
  &\leq \eps^{p(p-1)+a[(p-1)^2-2]} \left\| \L  W_\eps^p(t,r) \right\|_\Linftr{p-1} \\
  &\leq \eps^{p(p-1)+a[(p-1)^2-2]} \left\| W_\eps(t,r) \right\|_\Linftx{p-1}^p
  \leq C \eps^{p(p-1)(1-a)+a[(p-1)^2-2]}
\end{split}
\end{equation}
where the power of $\eps$ is again positive for any $a>0$.

Finally, we have to estimate the most important ``error'' term
\begin{equation}
  \Delta_\mu = \left| \L \eps^{-a (p+1)} u_0^p(\eps^{-a} t, \eps^{-a} r) - \L \mu \right|
\end{equation}
which will tend to zero as $\eps\ra 0$ because the compactly supported outgoing wave $u_0$ shrinks, under scaling, to the distribution $\mu$ localized on the lightcone $t=r$.
First, observe that only the outgoing part of the radiation contained in $u_0$ contributes to the  norm $\left\|\Delta_\mu\right\|_\Linftr{p}$ since
\begin{equation}
\begin{split}
  &\left| \L \eps^{-a(p+1)} u_0^p(\eps^{-a} t, \eps^{-a} r) - \L \mu \right|
  = \left| \L \eps^{-a} \frac{[h(\eps^{-a}(t-r)) - h(\eps^{-a}(t+r))]^p}{r^p} - \L \mu \right| \\
  &\leq \Big| \L \eps^{-a} \frac{h^p(\eps^{-a}(t-r))}{r^p} - \L \mu \Big|
  + \left| \L \eps^{-a} \frac{h^p(\eps^{-a}(t+r))}{r^p} \right|
\end{split}
\end{equation}
and the last term (ingoing wave) vanishes when $t-r>1$. The reason is that in the (spherically symmetric) Duhamel formula for $\L F$ with $F(t,r)=\eps^{-a}{h^p(\eps^{-a}(t+r))}/{r^p}$
\begin{equation}
  \L F(t,r) =
  \frac{1}{r}\int_{t-r}^{t+r} du \int_{-u}^{+u} dv\ (u-v) \frac{\eps^{-a}h^p(\eps^{-a} u)}{(u-v)^p}
\end{equation}
the integration runs only over $u>t-r$ and hence $\eps^{-a} u > R$ is outside of the support of $h$.

So we need to deal only with the first term (outgoing wave). We define an auxiliary function
\begin{equation}
  \Phi_p(t,r,\rho) := \L \left[ h^p(\rho) \frac{\delta(t-r)}{r^p} \right]
\end{equation}
and observe that
\begin{equation}
\begin{split}
  &\delta_\mu := \L \eps^{-a} \frac{h^p(\eps^{-a}(t-r))}{r^p} - \L \mu \\
  &= \eps^{-a} \int \L \left[ h^p(\eps^{-a}\rho) \frac{\delta(t-\rho-r)}{r^p} \right]  d\rho
   - \eps^{-a} \int \L \left[ h^p(\eps^{-a}\rho) \frac{\delta(t-r)}{r^p} \right]  d\rho\\
  &= \eps^{-a} \int \Phi_p(t-\rho,r,\eps^{-a}\rho)\  d\rho
   - \eps^{-a} \int \Phi_p(t,r,\eps^{-a}\rho)\ d\rho\\
  &= \eps^{-a} \int d\rho \int_{t-\rho}^t d\tau \ \d_\tau \Phi_p(\tau,r,\eps^{-a}\rho) \\
  &= -\eps^{-a} \int d\rho \ h^p(\eps^{-a}\rho) \int_{t-\rho}^t d\tau \
  \frac1{r} \left[ \frac{1}{(\tau-r)^{p-1}} - \frac{1}{(\tau+r)^{p-1}}\right]
\end{split}
\end{equation}
where in the last step we differentiated the explicit expression for $\Phi_p$
\begin{equation}
  \Phi_p(t,r,\rho) =  \frac{h^p(\rho)}{p-2} \cdot \frac1{r} \left[ \frac{1}{(t-r)^{p-2}} - \frac{1}{(t+r)^{p-2}}\right].
\end{equation}
Now, we can estimate
\begin{equation}
  |\delta_\mu| \leq \eps^{-a} \int_{-\eps^a R}^{+\eps^a R} d\rho \ |h^p(\eps^{-a}\rho)|
  \int_{t-\rho}^t d\tau \ \frac1{r} \left| \frac{1}{(\tau-r)^{p-1}} - \frac{1}{(\tau+r)^{p-1}} \right|
\end{equation}
The inner integral can be bound
\begin{equation}
\begin{split}
  &\int_{t-\rho}^t d\tau\ \frac1{r} \left| \frac{1}{(\tau-r)^{p-1}} - \frac{1}{(\tau+r)^{p-1}} \right|
  \leq \rho\sup_{t-\rho\leq\tau\leq t} \frac1{r} \left|\frac{1}{(\tau-r)^{p-1}}-\frac{1}{(\tau+r)^{p-1}}\right|\\
  &= \frac{\rho}{r} \left|\frac{1}{(t-r-\rho)^{p-1}}-\frac{1}{(t+r-\rho)^{p-1}}\right|
  \leq \frac{\rho}{r} \left|\frac{1}{(t-r-\eps^a R)^{p-1}}-\frac{1}{(t+r-\eps^a R)^{p-1}}\right|\\
  &\leq \rho \frac{C}{(t+r-\eps^a R)(t-r-\eps^a R)^{p-1}}
  \leq \rho \frac{C'}{\norm{t+r}\norm{t-r}^{p-1}}
\end{split}
\end{equation}
where we have used the facts that $\rho\in[-\eps^a R,+\eps^a R]$ and $t\pm r> 1> \eps^a R$ together with a simple algebraic inequality of Bernoulli's type (cf. proof of Lemma 1 in \cite{NS-DecayLemma}). This leads to
\begin{equation}
  |\delta_\mu| \leq
  \frac{\eps^{-a}C'}{\norm{t+r}\norm{t-r}^{p-1}}\int_{-\eps^a R}^{+\eps^a R}d\rho\ \rho\ |h^p(\eps^{-a}\rho)|
  = \frac{\eps^{a}C' C_{p,1}}{\norm{t+r}\norm{t-r}^{p-1}}
\end{equation}
for $t-r>1$ where
\begin{equation}
  C_{k,i} := \int \rho^i |h^k(\rho)| d\rho.
\end{equation}
Taking the $\Linftr{p}$ norm gives
\begin{equation}
  \left\| \Delta_\mu \right\|_\Linftr{p} =\left\| \delta_\mu \right\|_\Linftr{p} \leq C C_{p,1} \eps^a
\end{equation}
and we see that $\Delta_\mu$ has faster decay than other error terms. In order to respect this fact we instead of using further the norm-estimate \eqref{W-Weps-norm} come back to the pointwise notation and find
\begin{equation}
\begin{split}
  | W_\eps - W|
  &\leq |\Delta_p| + |\Delta_O| + |\Delta_W| + |\Delta_\mu| \\
  &= \OO[p-1]{\eps^{(p-1)(1-a)+a[(p-1)^2-2]/p}} \\
  &+ \OO[p-1]{\eps^{p(p-1)(1-a)+a[(p-1)^2-2]+pk} + \eps^{k(1+a)}}\\
  &+ \OO[p-1]{\eps^{p(p-1)(1-a)+a[(p-1)^2-2]}} \\
  &+ \OO[p]{C^1_p \eps^a}.
\end{split}
\end{equation}
We see that all terms tend to zero for $\eps\ra 0$ when $a>0$ and $0<a<p(p-1)/(p+1)$.
This result can be written as
\begin{equation}
  W_\eps(t,r) = W(t,r) + \OO[p-1]{\eps^{\la_0}} + \OO[p]{\eps^a}
\end{equation}
as $\eps\ra 0$ uniformly in the region $t-r>1$, where $\la_0>0$ depends on the values of $a$ and $k$
\begin{equation}
  \la_0 := \left\{\begin{array}{ll}
         \frac{p(p-1)(1-a)+a[(p-1)^2-2]}{p}, & \text{for } a>\frac{p(p-1-k)}{p+1+pk}, \\
	 k(1+a), & \text{otherwise.}
         \end{array}\right.
\end{equation}
The exact solution of \eqref{BoxW} reads
\begin{equation}
  W(t,r) = A_p \frac{\Theta(t-r)}{r} \left[ \frac{1}{(t-r)^{p-2}} - \frac{1}{(t+r)^{p-2}} \right]
\qquad \forall (t,r)\in\RTR
\end{equation}
where $A_p:= 2^{p-3} C_p / (p-2)$.
Finally, scaling all functions back, we find
\begin{equation}
\begin{split}
  u(t,r) &= \eps u_0(t,r) + \eps^{p+a(p-1)} W_\eps( \eps^a t, \eps^a r) \\
  &= \eps u_0(t,r) +
  \eps^p A_p \frac{\Theta(t-r)}{r} \left[ \frac{1}{(t-r)^{p-2}} - \frac{1}{(t+r)^{p-2}} \right]
  + \OO[p-1]{\eps^{p+\la_0}} + \OO[p]{\eps^p}
\end{split}
\end{equation}
where the error bound is now uniform in the region $t-r>\eps^{-a}$.
Observe that in that region we have $\OO[p]{\eps^p}=\OO[p-1]{\eps^{p+a}}$ so the total error can be written as $\OO[p-1]{\eps^{p+\la}}$ where $\la:=\min\{\la_0, a\}$
what finishes the proof of Theorem \ref{Th:main}.
\end{proof}

\xnewpage
\section{Discussion}

In the case when $F(u)=u^p$ the result is
\begin{equation}
\begin{split}
  u(t,r) &= \eps u_0(t,r) +
  \eps^p A_p \frac{\Theta(t-r)}{r} \left[ \frac{1}{(t-r)^{p-2}} - \frac{1}{(t+r)^{p-2}} \right]
  + \OO[p-1]{\eps^{p+\la_0}} + \OO[p]{\eps^p}
\end{split}
\end{equation}
with $\la_0:=[p(p-1)(1-a)+a((p-1)^2-2)]/p$. The first error term describes corrections with the same decay in time as the leading order asymptotics but entering with higher powers of $\eps$ while the second error term stays for a correction entering at the same nonlinear order (the same power of $\eps$) but having faster decay in time.
By a suitable choice of $a$, and thus the asymptotic region of spacetime, one can make the first or the second error term dominant and explicitly evaluate its leading asymptotics from initial data.
Hence, the technique can be extended to study the sub-leading asymptotics order by order in the powers of $\eps$. By defining another auxiliary function $v:=u-\eps u_0 - \eps^p w$ and scaling it properly and taking the limit $\eps\ra 0$ one can obtain another linear wave equation for $v$ with a distributional source (containing now $\delta'(t-r)$). This procedure can be, in principle, iteratively repeated to arbitrarily high orders in $\eps$, however the number of terms to be analyzed in the limit $\eps\ra 0$ quickly increases. In a forthcoming publication we want to apply this generalization to the well-studied cubic wave equation and derive the \textit{late-time attractor} found by Bizon and Zenginoglu in \cite{Bizon+Anil}.

In order to generalize this result from spherical symmetry to 3-dimensions one additionally needs to show that the new $u_0$, as a solution to the linear equation \eqref{free-wave} with the initial data \eqref{free-wave-init}, is near to the analogue of $u_0$ given by \eqref{free-u0}, containing now the Friedlander's radiation field (for the details see \cite[Lemmas 2.3-2.5]{Lindblad-PhD_CPDE}). This step does not seem to essentially modify the rest of the proof except that spherical means will appear at various places and the constant $A_p$ will include integration over spheres, too.

As already mentioned, the scaling technique developed here is not by itself restricted to the semilinear wave equations of the form \eqref{wave-eq}. It can be easily generalized to any system of nonlinear wave equations. However, the strength of the estimate on the error term will crucially depend on the availability of corresponding pointwise decay estimates.





\begin{acknowledgments}
I would like to thank Hans Lindblad for introducing me to the technique of scaling developed in his PhD Thesis and for interesting discussions we had during my stay at the University of California in San Diego.

I also kindly acknowledge the support and hospitality of the Institute Mittag-Leffler in Stockholm and the Institute Henri Poincar{\'e}, Centre Emile Borel in Paris where parts of this work have been done.
\end{acknowledgments}

\bibliography{QNMs}
\bibliographystyle{unsrt}

\end{document}